\documentclass[runningheads]{llncs}       

\usepackage{makeidx,enumerate,graphicx,amsfonts,amssymb,amsmath}
\usepackage{graphicx,epsfig}
\usepackage{fontenc}
\usepackage[utf8]{inputenc}
\usepackage{mathtools}
\usepackage{appendix}
\usepackage{algorithm2e}
\usepackage{url}
\usepackage{caption}
\usepackage{longtable}
\usepackage{xcolor}
\usepackage{pgffor}
\usepackage{bbm}
\usepackage{pgffor}
\usepackage{bbm}
\usepackage{todonotes}

\foreach \x in {A,...,Z}{%
	\expandafter\xdef\csname b\x\endcsname{\noexpand\ensuremath{\noexpand\mathbf{\x}}}
	\expandafter\xdef\csname c\x\endcsname{\noexpand\ensuremath{\noexpand\mathcal{\x}}}
	\expandafter\xdef\csname B\x\endcsname{\noexpand\ensuremath{\noexpand\mathbb{\x}}}
}

\foreach \x in {a,...,z}{%
	\expandafter\xdef\csname b\x\endcsname{\noexpand\ensuremath{\noexpand\mathbf{\x}}}
}

\newcommand{\ab}{\allowbreak}

\newcommand{\xc}{\mathop{\mathrm{xc}}} 
\newcommand{\size}{\mathop{\mathrm{size}}} 

\newcommand{\conv}{\mathrm{conv}}
\newcommand{\qappolytope}[1]{\ensuremath{\mathrm{QAP}_{#1}}}
\newcommand{\qaprelax}[2]{\ensuremath{\mathrm{R}^{#1}_{#2}}}

\newcounter{qapno}
\DeclareRobustCommand{\qaplabel}[1]{%
	\refstepcounter{qapno}%
	\label{#1}}

\begin{document}

\title{On the Complexity of Some Facet-Defining Inequalities of the QAP-polytope}

\author{Pawan Aurora\inst{1} \and
	Hans Raj Tiwary\inst{2}
}

\institute{IISER Bhopal, India \email{paurora@iiserb.ac.in} \and
	Department of Applied Mathematics, Charles University, Prague, Czech Republic \email{hansraj@kam.mff.cuni.cz}
}
	\maketitle
	
\begin{abstract}
	The Quadratic Assignment Problem (QAP) is a well-known NP-hard problem that is equivalent to optimizing a linear objective function 
	over the QAP polytope. The QAP polytope with parameter $n$ - \qappolytope{n} - is defined as the convex hull of rank-$1$
	matrices $xx^T$ with $x$ as the vectorized $n\times n$ permutation matrices. 
	
	In this paper we consider all the known exponential-sized families of facet-defining inequalities of the QAP-polytope. We describe a new family of valid inequalities that we show to be facet-defining. We also show that 
	membership testing (and hence optimizing) over some of the known classes of inequalities is coNP-complete. We complement our 
	hardness results by showing a lower bound of $2^{\Omega(n)}$ on the extension complexity of all relaxations of \qappolytope{n} for which any of the known 
	classes of inequalities are valid.
\end{abstract}
	
\section{Introduction}
\label{sec:intro}
The Quadratic Assignment Problem (QAP) is a fundamental combinatorial optimization problem from the category of facility location 
problems \cite{KoopmansB55,Lawler63}. QAP is defined as the following problem: given $n$ facilities and $n$ locations, distances
$d_{ij}$ between all pairs of locations $i,j\in[n]$, flows $f_{ij}$ between all pairs of facilities $i,j\in[n]$ and costs $c_{ij}$ 
of opening facility $i$ at location $j$, for all pairs $i,j\in[n]$, find an assignment $\sigma$ of the $n$ facilities to the $n$ 
locations so that the total cost given by the function $\sum_{i,j}f_{ij}d_{\sigma(i)\sigma(j)}+\sum_ic_{i\sigma(i)}$ is minimized.
The problem is known as QAP since it can be modeled as optimizing a quadratic function over linear and binary constraints. However, 
several linearizations of the problem have been proposed. For details refer to the book \cite{Cela1997} and the citations therein.

Given an instance of QAP, it is NP-hard to approximate the optimum within any constant 
factor \cite{SahniG1976}. What makes QAP one of the ``hardest'' problems in combinatorial optimization 
is the fact that unlike most NP-hard combinatorial optimization problems, it is practically intractable. It is generally considered
impossible to solve to optimality QAP instances of size larger than $20$ within reasonable time limits \cite{Cela1997}.

As is common with combinatorial obtimization problems, QAP can be viewed as the problem of optimizing a linear objective function 
over the convex hull of all feasible solutions. To this end, the QAP polytope is defined
as $\displaystyle \qappolytope{n}=conv\left(\left\{yy^T|y=vec(P_{\sigma}),\sigma\in S_n\right\}\right)$, where $P_{\sigma}$ is the $n\times n$ permutation matrix corresponding to the premutation $\sigma$ and $y=vec(P_{\sigma})$
is its vectorization. Following the notation of \cite{AuroraM18}, we denote a vertex $yy^T$ as $P^{[2]}_{\sigma}$. Note that 
$P^{[2]}_{\sigma}(ij,kl)=P_{\sigma}(i,j)\cdot P_{\sigma}(k,l)$. Clearly, $\qappolytope{n}\subset\mathbb{R}^{n^2\times n^2}$. In fact $\qappolytope{n}$ can be embedded in
$\mathbb{R}^{(n^4+n^2)/2}$ since each point in the polytope is a symmetric $n^2\times n^2$ matrix and we could only store its upper
(or lower) triangular part. However, in this paper we would conveniently denote a point in \qappolytope{n} by a $n^2\times n^2$ 
matrix.

One of the methods for solving hard combinatorial optimization problems is the method of branch-and-cut \cite{PadbergR1991}. For this 
method to be effective for the QAP, it is important to identify new valid and possibly facet-defining inequalities for the QAP-polytope
and to develop the corresponding separation algorithms.
Given that it is NP-hard to optimize over the QAP-polytope, it is probably impossible to characterize all its facets \cite{Pit:prob89}. In 
\cite{JungerK1996a,JungerK1996b,PadbergR1996}, the authors obtain early results on the combinatorial structure of the QAP-polytope 
and some of its facet-defining inequalities. 
In \cite{AuroraM18} the authors list all the known facets of the QAP-polytope besides the equations that define its affine hull. 
In this paper we add another exponential sized family to the list of known facets of the QAP-polytope. Optimizing the QAP objective function over any of the relaxations
given by these families can provide an approximate solution to the QAP, provided the optimization problem can be efficiently solved. In this paper we also show
that optimizing over the relaxations given by some of these exponential sized family of facet-defining inequalities is NP-hard. We do it by proving that the corresponding membership testing problem is coNP-complete for the appropriate classes of inequalities. 

Furthermore, we prove a lower bound of $2^{\Omega(n)}$ on the extension complexity of bounded relaxations of \qappolytope{n} obtained by each of these families of 
inequalities. 

\medskip
To summarize, our main contributions are as follows. 
\begin{itemize}
	\item We identify a new family of valid inequalities for \qappolytope{n} (Section \ref{sec:relaxations}) and prove that they are facet-defining (Section \ref{sec:newfacets}),	
	\item We prove that membership testing for three out of the five known families of valid inequalities for the QAP-polytope(including the new one we introduce) is coNP-complete (Section \ref{sec:sepncomp}), and
	\item We prove a lower bound of $2^{\Omega(n)}$ for the extension complexity of any bounded\footnote{In fact, boundedness is not required for the results in Section \ref{sec:extncomp}. However, since we will rely on existing results, such as Theorem \ref{thm:extformcc}, that are published with the boundedness assumption, we will include this assumption.} relaxation of \qappolytope{n} that has any of the known families as valid inequalities (Section \ref{sec:extncomp}).
\end{itemize}

\subsection{Extension Complexity}
Let $P\subset\BR^n$ be a polytope. A polytope $Q\subset\BR^{n+r}$ is called an \emph{extension} or an \emph{extended formulation} of $P$ if $$P=\{x\in\BR^n~|~\exists y\in\BR^r, (x,y)\in Q\}.$$
Let $\size(P)$ denote the number of facets of polytope $P$ and let $Q\downarrow P$ denote that $Q$ is an extended formulation of $P$. Then, the extension complexity of a polytope $P$ - denoted by $\xc(P)$ - is defined to be $\displaystyle\min_{Q\downarrow P}\size(Q)$.

Extended formulations are a very useful tool in combinatorial optimization as they allow the possibilty of drastically reducing the size of a Linear Program by introducing new variables (See \cite{VanderbeckWolsey2010,Wolsey11,Kaibel11,anor/ConfortiCZ13} for surveys). In the past decade lower bounds on the extension complexity of various polytopes have been studied \cite{mp/Rothvoss13,jacm/FioriniMPTW15,mp/AvisT15,jacm/Rothvoss17} and the notion generalized and studied in various settings such as general conic extensions \cite{mor/GouveiaPT13}, semidefinite extensions \cite{mp/BrietDP15,stoc/LeeRS15}, approximation \cite{BM13,CLRS13,mor/BraunFPS15,mor/BazziFPS19}, parameterization \cite{orl/Buchanan16,dam/GajarskyHT18}, generalized probabilistic theories in Physics \cite{Fiorini_2014}, and information theoretic perspective \cite{BM13,cc/BraunJLP17}.

Superpolynomial lower bounds on extension complexity are known for polytopes related to many NP-hard problems \cite{jacm/FioriniMPTW15,mp/AvisT15} as well as for the Matching polytope: the convex hull of characteristic vectors of all matchings in $K_n$ \cite{jacm/Rothvoss17}. High extension complexity of the Matching polytope highlights the fact that a superpolynomial lower bound on the extension complexity  cannot be taken to mean that the underlying optimization problem is not solvable in polynomial time. However, these lowers bounds are unconditional and do not require standard complexity theoretic assumption such as $P\neq NP$. Moreover, apart from the exception of Matching polytope, linear optimization over all known polytopes with superpolynomial lower bound is infeasible. Either because the linear optimization over the polytope is NP-hard \cite{jacm/FioriniMPTW15,mp/AvisT15} or the polytope is not explicitly given, as is the case for some matroid polytopes \cite{mp/Rothvoss13}.

For the purposes of this paper the most relevant characterization of extension complexity is given by Faenza et al. \cite{mp/FaenzaFGT15} where the authors prove the equivalence between existence of an extended formulation of size $r$ with the existence of a certain two-party communication game requiring an exchange of $\Theta(\log{r})$ bits. We will describe this connection here and use it as a black box in our proofs of lower bounds on the extension complexity of the polytopes considered here.

\subsubsection{EF-protocols: Computing a matrix in expectation}
Let $\mathrm{M}$ be an $m\times n$ matrix with non-negative entries. Consider a communication game between two players: Alice and Bob. Alice and Bob both know the matrix $\mathrm{M}$ and can agree upon any strategy prior to the start of the game. In each round of the game, Alice receives a row index $i\in[m]$ and Bob a column index $j\in[n]$. Both Alice and Bob have no restriction on the computations that they perform and can also use (private) random bits. They can also exchange information by sending some bits to the other player. At some point one of them outputs a non-negative number and the round finishes.

Since they are allowed the use of random bits, the output $X_{ij}$ when Alice and Bob receive inputs $i$ and $j$ respectively, is a random variable. We says that their strategy is an \emph{EF-protocol} for $\mathrm{M}$ if $\BE[X_{ij}]=M_{ij}$ for all $(i,j)\in[m]\times[n]$, where $\BE[X_{ij}]$ is the expected value of the random variable $X_{ij}$. The complexity of an EF-protocol is defined to be the maximum number of bits exchanged between Alice and Bob for any input $i,j$ to Alice and Bob respectively.

Let $P=\{x~|~Ax\leqslant b\}=\conv(\{v_1,\ldots,v_n\})$ be a polytope where $A$ is an $m\times d$ real matrix, $b\in\BR^m,$  $v_i\in\BR^d$, and $\conv(S)$ denotes the convex hull of the points in a set $S$. The slack matrix of $P$ with respect to this representation - denoted by $S(P)$ - is the $m\times n$ (non-negative) matrix whose entry at $i$-th row and $j$-th column is $b_i-A_iv_j$, where $A_i$ denotes the $i$-th row of the matrix $A$. Note that a polytope is not defined uniquely this way: one can always embed the polytope in higher dimensional space, and add redundant inequalities and points to the descriptions. However, in what follows, none of that makes any difference and one can choose any description that they like. This justifies the notation $S(P)$ for any slack matrix of $P$ even though the particular description of $P$ that defines this matrix is completely ignored in the notation.

The following connection -- which we will use in Section \ref{sec:extncomp} -- was shown by Faenza et al. \cite{mp/FaenzaFGT15} between existence of an EF-protocol computing a slack matrix of polytope $P$ and that of an extended formulation of $P$.

\begin{theorem}\cite{mp/FaenzaFGT15}\label{thm:extformcc}
	Let $M$ be a non-negative matrix such that any EF-protocol for $M$ has complexity at least $c$. Further, let $P$ be a polytope such that $M$ is a submatrix of some slack matrix $S(P)$ of $P$. Then, $\xc(P)\geqslant 2^c$.
\end{theorem}


\section{Relaxations of \qappolytope{n}}\label{sec:relaxations}
In the following, $Y$ is a $n^2\times n^2$ variable matrix that is used to denote an arbitrary point in \qappolytope{n}. Further, $Y_{ij,kl}$
refers to $Y(n*(i-1)+j,n*(k-1)+l)$.

The most general family of valid inequalities that includes all known families as special cases is the following:

\medskip
\noindent\textbf{QAP\ref{qap5}:}\qaplabel{qap5} 
\vspace*{-1.5\baselineskip}\begin{equation}\label{qap5ineq}
\displaystyle\sum_{ijkl}n_{ij}n_{kl}Y_{ij,kl} - (2\beta-1)\sum_{ij}n_{ij}Y_{ij,ij}\geqslant \frac{1}{4}-(\beta-\frac{1}{2})^2
\end{equation}
where $\beta\in \mathbb{Z}$ and $n_{ij}\in \mathbb{Z}$ for all $i,j\in[n]$.
These inequalities were introduced in \cite{AuroraM18} as a generalization of all known facet-defining inequalities for the QAP-polytope. 

\medskip\noindent\textbf{QAP\ref{qap1}:}\qaplabel{qap1}
\vspace*{-1.5\baselineskip}\begin{equation}\label{qap1ineq}
\begin{tabular}{rcl}
$\displaystyle\sum_{r=1}^{m}Y_{i_rj_r,kl} - Y_{kl,kl}- \sum_{r < s}Y_{i_rj_r,i_sj_s}$ & $\leqslant$ & $0$
\end{tabular}
\end{equation}

where $i_1,\dots,i_m,k$ are all distinct and $j_1,\dots,j_m,l$ are also
distinct. In addition, $n\geqslant 6, m\geqslant 3$.
These inequalities were introduced in \cite{AuroraM18} and proved to be a special case of QAP\ref{qap5} that are facet-defining for the QAP-polytope. 

\noindent\textbf{QAP\ref{qap2}:}\qaplabel{qap2} 
\vspace*{-0.8\baselineskip}\begin{equation}\label{qap2ineq}
\begin{tabular}{rcl}
$\displaystyle(\beta-1)\sum_{(ij)\in P\times Q}Y_{ij,ij}-\sum_{\substack{(ij),(kl)\in P\times Q\\i<k}}Y_{ij,kl}$ & $\leqslant$ & $\frac{\beta^2-\beta}{2}$
\end{tabular}
\end{equation}
where $P,Q\subset[n]$. In addition (i) $\beta +1 \leqslant\lvert P\rvert,\lvert Q\rvert\leqslant n-3$, (ii) $\lvert P\rvert+\lvert Q\rvert\leqslant 
n-3+\beta$, (iii) $\beta\geqslant2$.
These inequalities were introduced by J\"unger and Kaibel \cite{JungerK2001,KaibelPhD97} who also proved that they are facet-defining for the QAP-polytope. They are also a special case of QAP\ref{qap5} \cite{AuroraM18}.

\noindent\textbf{QAP\ref{qap3}:}\qaplabel{qap3} 
\begin{equation}\label{qap3ineq}
\begin{tabular}{rrrrrrl}
$-$ & $\displaystyle(\beta-1)\sum_{(ij)\in P_1\times Q}Y_{ij,ij}$ & $+$ & $\displaystyle\beta \sum_{(ij)\in P_2\times Q}Y_{ij,ij} $ & 
$+$ & $\displaystyle\sum_{\substack{i<k \\ (ij),(kl)\in P_1\times Q}}Y_{ij,kl}$ \\ 
$+$ &$\displaystyle\sum_{\substack{i<k \\ (ij),(kl)\in P_2\times Q}} Y_{ij,kl}$ & 
$-$ & $\displaystyle\sum_{\substack{(ij)\in P_1\times Q \\ (kl)\in P_2\times Q}}Y_{ij,kl}$ & $-$ & $ \frac{\beta^2-\beta}{2}$ & $\geqslant 0$
\end{tabular}
\end{equation}
where $P_1,P_2,Q\subset[n],P_1\cap P_2=\emptyset$. Further, (i) $3\le|Q|\le n-3$, (ii) $|P_1|+|P_2|\le n-3$, (iii) $|P_1|\geqslant\min\{
2,\beta+1\}$, (iv) $|P_2|\geqslant\min\{1,-\beta+2\}$, (v) $\lvert|P_1|-|P_2|-\beta\rvert\le n-|Q|-4$, (vi) if $|P_2|=1$: 
$|Q|\geqslant\min\{-\beta+5,\beta+2\}$; if $|P_2|\geqslant2$: $|Q|\geqslant\min\{-\beta+5,\beta+3\}$ or $|Q|\geqslant\min\{-\beta+4,\beta+4\}$.
These inequalities were also introduced by Kaibel \cite{KaibelPhD97} and shown to be facet-defining for the QAP-polytope. They are 
also a special case of QAP\ref{qap5}.

\medskip

It is known \cite{AuroraM18} that the inequalities QAP\ref{qap1} , QAP\ref{qap2} and QAP\ref{qap3} are special instances of 
the QAP\ref{qap5} inequalities. QAP\ref{qap5} inequalities are in general not facet-defining for the QAP-polytope and so it is interesting to identify conditions under which they do define facets. We identify a new special case (\ref{qap4ineq}) of QAP\ref{qap5} inequality and show in Section \ref{sec:newfacets} that they are facet-defining. Inequality QAP\ref{qap4}
follows from QAP\ref{qap5} by setting $\beta=2$, $n_{i_1j_1}=n_{i_2j_2}=\dots =n_{i_m,j_m}=1$ for distinct $i_1,\dots,i_m\in[n]$ 
and distinct $j_1,\dots,j_m\in[n]$, and $n_{ij}=0$ for $i\in[n]\setminus\{i_1,\ldots,i_m\}$ or $j\in[n]\setminus\{j_1,\ldots,j_m\}$.

\medskip
\noindent\textbf{QAP\ref{qap4}:}\qaplabel{qap4} 
\vspace*{-1.5\baselineskip}\begin{equation}\label{qap4ineq}
\begin{tabular}{rcl}
$\displaystyle\sum_{r=1}^{m}Y_{i_rj_r,i_rj_r}-\sum_{r < s}Y_{i_rj_r,i_sj_s}$ & $\leqslant$ & $1$
\end{tabular}
\end{equation}
where $i_1,\dots,i_m$ are all distinct and $j_1,\dots,j_m$ are also
distinct. In addition, $m,n\geqslant 7$. These inequalities are new and we discuss them in the following section.


\medskip
\section{A new class of facet-defining inequalities}
\label{sec:newfacets}

In this section we prove that the inequalities QAP\ref{qap4} are facet-defining. Let $S_k$ denote the set of 
those vertices of $\qappolytope{n}$ that correspond to the permutations having $i_r\mapsto j_r$ for $r\in\{r_1,r_2,\ldots,r_k\}\subseteq[m]$ 
and $i_r\not\mapsto j_r$ for $r\in[m]\setminus\{r_1,r_2,\ldots,r_k\}$. Here $i_r,j_r,m$ are as in the definition of QAP\ref{qap4}. 
If $V$ is the set of all the vertices of $\qappolytope{n}$ then clearly $V=\cup_{i=0}^{m}S_i$, and $S_i\cap S_j=\emptyset$ for all 
$i\neq j\in\{0,1,\ldots,m\}$.
\begin{lemma}\label{lem:s1s2eqn}
	The sets $S_1,S_2$ together constitute the vertices that satisfy the inequality (\ref{qap4ineq}) with equality.
\end{lemma}
\begin{proof}
	Consider a $P^{[2]}_{\sigma}\in S_k$ for $1\leqslant k\leqslant m$. W.l.o.g. let $\sigma(i_1)=j_1,\ldots,\sigma(i_k)=j_k$ 
	and $\sigma(i_r)\neq j_r$ for $r\in\{k+1,\ldots,m\}$. Substituting $Y=P^{[2]}_{\sigma}$ in (\ref{qap4ineq}) we have 
	$\sum_{r=1}^{m}P_{\sigma}(i_r,j_r)-\sum_{r < s}P_{\sigma}(i_r,j_r)\cdot P_{\sigma}(i_s,j_s)=k-\binom{k}{2}$. For $k=0$ we have
	$k-\binom{k}{2}=0<1$, for $k=1,k-\binom{k}{2}=1$, for $k=2,k-\binom{k}{2}=1$ and for $3\leqslant k\leqslant m$ we have
	$k-\binom{k}{2}<1$. Hence a vertex of $\qappolytope{n}$ satisfies the inequality (\ref{qap4ineq}) with equality if and only
	if it belongs to $S_1$ or $S_2$.
	\qed
\end{proof}

Let $S=S_1\cup S_2$. We will show that any vertex in $V\setminus S$ can be expressed as a linear combination of the vertices in $S$ 
and a fixed vertex $P^{[2]}_{\sigma^*}\in S_0$. This will establish that the dimension of the face containing $S$ is one less than 
the dimension of the polytope and hence it must be a facet.

The following lemma from \cite{AuroraM18} provides a useful tool to express certain vertices as a linear combination of others.

\begin{lemma}\label{lem:identity1}\cite[Lemma 16]{AuroraM18}
	Let $k_1,k_2,k_3,x,y\in [n]$ be distinct indices. Let $\Sigma=\{\sigma_1,\ldots\ab,\sigma_6\}$ be a set of permutations of 
	$[n]$ such that $\sigma_i(z)=\sigma_j(z)$ for all $z\in [n]\setminus\{k_1,k_2,k_3\}$ and for every $i,j\in \{1,\dots,6\}$. 
	Further, let $\Sigma'=\{\sigma'_1,\ldots,\sigma'_6\}$ where $\sigma'_i$ is a transposition of $\sigma_i$ on the indices 
	$x,y$, for each $i=1,\dots,6$. Then $\forall\; i,j,k,l\in[n],\ \sum_{\sigma\in\Sigma\cup\Sigma'}sign(\sigma)P^{[2]}_{\sigma}
	(ij,kl)=0$.
\end{lemma}

The next lemma shows that the vertices in the sets $S_4,\ab\ldots,S_m$ can be expressed as a linear combination of the 
vertices in the sets $S_0,\ldots,S_3$. 

In what follows, when it is clear from the context, we use $\sigma$ to refer to 
a vertex $P^{[2]}_{\sigma}$.

\begin{lemma}\label{lem:skasnxt4}
	For $k\geqslant 4$, any vertex in $S_k$ can be expressed as a linear combination of vertices in $S_{k-1},S_{k-2},S_{k-3},$ 
	and $S_{k-4}$.
\end{lemma}
\begin{proof}
	Let $\sigma_1\in S_k$. Since $k$ is at least $4$, we must have indices $i_{r_1},i_{r_2},i_{r_3},i_{r_4},r_k\ab\in[m],k=1,\ldots,
	4$, such that $\sigma_1(i_{r_k})=j_{r_k},k=1,\ldots,4$. Let $k_1=i_{r_1},k_2=i_{r_2},k_3=i_{r_3},x=i_{r_4}$, where $k_1,k_2,
	k_3,x$ are as defined in Lemma \ref{lem:identity1}. Applying Lemma \ref{lem:identity1} with $y$ chosen as an index such that
	either $y\neq i_p$ for any $p\in[m]$ or when $y=i_p,p\in[m]$ then $\sigma_1(i_p)\neq j_p$, we get $\sigma'_1\in S_{k-1},
	\sigma_2,\sigma_3,\sigma_6\in S_{k-2},\sigma_4,\sigma_5,\sigma'_2,\sigma'_3,\sigma'_6\in S_{k-3},\sigma'_4,\sigma'_5\in 
	S_{k-4}$ with the property that $\sigma_1$ is a linear combination of $\sigma_2,\ldots,\sigma_6,\sigma'_1,\ldots,\sigma'_6.$
	\qed
\end{proof}

Next, we show that vertices in $S_3$ can be expressed as  linear combinations of vertices in $S$ and $S_0$ as well.

\begin{lemma}\label{lem:s3ss0}
	Any vertex in $S_3$ can be expressed as a linear combination of vertices in $S$ and $S_0$.
\end{lemma}
\begin{proof}
	Let $\sigma_1\in S_3$. So we have indices $i_{r_1},i_{r_2},i_{r_3},r_k\in[m],k=1,2,3$, such that $\sigma_1(i_{r_k})\ab=
	j_{r_k},k=1,2,3$. Let $k_1=i_{r_1},k_2=i_{r_2},x=i_{r_3}$, where $k_1,k_2,x$ are as defined in Lemma \ref{lem:identity1}.
	Applying Lemma \ref{lem:identity1} with $k_3,y$ chosen arbitrarily from $[n]\setminus\{i_{r_1},i_{r_2},i_{r_3}\}$, we get 
	$\sigma_2,\ldots,\sigma_6,\sigma'_1,\sigma'_2,\sigma'_6\in S,\sigma'_3,\sigma'_4,\sigma'_5\in S_0$ with the property that $\sigma_1$ is a linear combination of $\sigma_2,\ldots,\sigma_6,\sigma'_1,\ldots,\sigma'_6$.
	\qed
\end{proof}

Now that we have established that the linear hull of $S_0,S_1,S_2$ equals the linear hull of the QAP-polytope, the only remaining task is to show that instead of the entire set $S_0$, a fixed vertex in $S_0$ suffices to generate the entire linear hull. We will show that in fact any arbitrary vertex in $S_0$ is sufficient. We do this by first showing that permutations in $S_0$ define a connected graph if the edges connect permutations that are one transposition apart. Then, we show that any vertex $\sigma\in S_0$ is sufficient to generate all vertices in the connected component of $\sigma$ by linear combination with vertices in $S_1, S_2$. 

In the following lemma we show that it is possible to obtain the permutation corresponding to a vertex in $S_0$ from the permutation
corresponding to any other vertex in $S_0$ via transpositions such that the vertices corresponding to the intermediate permutations
also lie in $S_0$.

\begin{lemma}\label{lem:szeroconn}
	Consider the graph $G=(S_0,E)$ where $S_0$ is the set of vertices of the QAP-polytope that correspond to permutations for 
	which $i_r\not\mapsto j_r$ for all $r\in[m]$ and $\{P^{[2]}_{\sigma_1},P^{[2]}_{\sigma_2}\}\in E$ for some 
	$P^{[2]}_{\sigma_1},P^{[2]}_{\sigma_2}\in S_0$ if $\sigma_1$ and $\sigma_2$ are transpositions of each other. Then 
	$G$ is connected.
\end{lemma}
\begin{proof}
	Consider an arbitrary vertex $P^{[2]}_{\sigma}\in S_0$ such that $\sigma(i_r)=k_r,k_r\neq j_r$ for all $r\in[m]$.
	Also, consider another vertex $P^{[2]}_{\sigma'}\in S_0$ such that $\sigma'(i_r)=l_r,l_r\neq j_r$ for all $r\in[m]$. We will
	show that there is a path from $P^{[2]}_{\sigma}$ to $P^{[2]}_{\sigma'}$ in $G$. For simplicity, we will use $\sigma$ to
	refer to $P^{[2]}_{\sigma}$. Let $\sigma,\sigma_1,\sigma_2,\ldots,\sigma_t,\sigma'$ be a path of length $t+1$ between 
	$\sigma$ and $\sigma'$. Consider a vertex $\sigma_p,p\in[t]$ such that $\sigma_p(i_r)=l_r$ for all $r\in[s],s<m$. In the
	next step we will extend the path from $\sigma_p$ to some vertex $\sigma_q$ such that $\sigma_q(i_r)=l_r$ for all $r\in[s+1]$.
	If $\sigma_p(i_x)=l_{s+1}$ such that $\sigma_p(i_{s+1})\neq j_x$ then we can swap $\sigma_p(i_x)$ with $\sigma_p(i_{s+1})$
	to get the desired vertex $\sigma_q$. Otherwise, in the first swap we can move $l_{s+1}$ to some index $i_{x'},x\neq x'$, 
	such that $l_{s+1}\neq j_{x'}$ and then in the second swap get $\sigma_p(i_{s+1})$ to map to $l_{s+1}$. Note that both the 
	swaps result in vertices within $S_0$. After the first swap we have $\sigma_{p'}(i_{x'})=l_{s+1}$ and $\sigma_{p'}(i_x)=
	\sigma_p(i_{x'})$ and after the second swap we get $\sigma_q(i_{s+1})=l_{s+1}$ and $\sigma_q(i_{x'})=j_x$, both of which
	avoid a map from $i_x$ to $j_x$ and $i_{x'}$ to $j_{x'}$. In case it is not possible to find a suitable $i_{x'}$ to move
	$l_{s+1}$, it should be possible to move $\sigma_p(i_{s+1})$ instead. Once we have obtained a permutation $\sigma''$ such 
	that $\sigma''(i_r)=l_r$ for all $r\in[m]$, there must exist a path from $\sigma''$ to $\sigma'$ since the set of 
	permutations having $i_r\mapsto l_r$ for all $r\in[m]$, forms a group isomorphic to the symmetric group on $n-m$ elements.
	\qed
\end{proof}

The following lemma gives a sequence of four vertices of $QAP_n$ such that a specific linear combination of these vertices reduces
the number of non-zero entries in the resulting vector to a constant independent of $n$. This lemma will be used crucially in Lemma
\ref{lem:szeroins} to express the difference of two neighboring vertices in $S_0$ in terms of the vertices in $S$.

\begin{lemma}\label{lem:identity2}
	Given a sequence of permutations over the set $[n]$, $\sigma_1,\sigma_2,\sigma_3,\sigma_4$, such that $\sigma_2$ is 
	obtained from $\sigma_1$ by a transposition that swaps the values of $\sigma_1(i),\sigma_1(j)$; $\sigma_3$ is 
	obtained from $\sigma_2$ by a transposition that swaps the values of $\sigma_2(i'),\sigma_2(j')$ ($i'\neq 
	i,j'\neq j$); and $\sigma_4$ is obtained from $\sigma_3$ by a transposition that swaps the values of 
	$\sigma_3(i),\sigma_3(j)$. Then $(P^{[2]}_{\sigma_1}-P^{[2]}_{\sigma_2})-(P^{[2]}_{\sigma_4}-P^{[2]}_{\sigma_3})$ 
	has a number of non-zeroes that is independent of $n$.
\end{lemma}
\begin{proof}
	Recall that $P^{[2]}_{\sigma}(ab,xy)=P_{\sigma}(a,b)\cdot P_{\sigma}(x,y)$. Since $\sigma_1$ and $\sigma_2$ differ at only 
	$i,j$, we have $(P^{[2]}_{\sigma_1}-P^{[2]}_{\sigma_2})(ab,xy)=0$ for all $a,b,x,y\in[n]\setminus\{i,j\}$. Similarly, we 
	have $(P^{[2]}_{\sigma_4}-P^{[2]}_{\sigma_3})(ab,xy)=0$ for all $a,b,x,y\in[n]\setminus\{i,j\}$. One can verify that 
	$(P^{[2]}_{\sigma_1}-P^{[2]}_{\sigma_2})(ab,xy)=1$ for all $x,y$ such that $\sigma_1(x)=y$, when $a=i,b=\sigma_1(i)$ or 
	when $a=j,b=\sigma_1(j)$. Symmetrically, we have $(P^{[2]}_{\sigma_1}-P^{[2]}_{\sigma_2})(ab,xy)=-1$ for all $x,y$ such 
	that $\sigma_2(x)=y$, when $a=i,b=\sigma_2(i)$ or when $a=j,b=\sigma_2(j)$. For the case when $a,b\notin\{i,j\}$, 
	$(P^{[2]}_{\sigma_1}-P^{[2]}_{\sigma_2})(ab,xy)=1$ when $x=i,y=\sigma_1(i)$ or when $x=j,y=\sigma_1(j)$ and 
	$(P^{[2]}_{\sigma_1}-P^{[2]}_{\sigma_2})(ab,xy)=-1$ when $x=i,y=\sigma_2(i)$ or when $x=j,y=\sigma_2(j)$. Similar values 
	follow for $P^{[2]}_{\sigma_4}-P^{[2]}_{\sigma_3}$. Note that $P^{[2]}_{\sigma_1}-P^{[2]}_{\sigma_2}$ and 
	$P^{[2]}_{\sigma_4}-P^{[2]}_{\sigma_3}$ differ only at the indices $i',j'$. So subtracting the latter from the former we 
	get, $((P^{[2]}_{\sigma_1}-P^{[2]}_{\sigma_2})-(P^{[2]}_{\sigma_4}-P^{[2]}_{\sigma_3}))(ab,xy)=0$ for all $a,b,x,y\notin
	\{i,j,i',j'\}$. The only non-zero entries that remain are the following: (i) $a=i,b=\sigma_1(i),x=i',y=\sigma_1(i')$, (ii) 
	$a=i,b=\sigma_2(i),x=i',y=\sigma_2(i')$, (iii) $a=j,b=\sigma_1(j),x=i',y=\sigma_1(i')$, (iv) $a=j,b=\sigma_2(j),x=i',
	y=\sigma_2(i')$, (v) $a=i,b=\sigma_1(i),x=i',y=\sigma_3(i')$, (vi) $a=i,b=\sigma_2(i),x=i',y=\sigma_3(i')$, (vii) $a=j,
	b=\sigma_1(j),x=i',y=\sigma_3(i')$, (viii) $a=j,b=\sigma_2(j),x=i',y=\sigma_3(i')$. Another $8$ non-zero entries correspond 
	to the case when $x=j'$ taking the total to $16$. $16$ more entries follow from symmetry, by swapping $a,b$ with $x,y$. 
	Thus, we get a total of $32$ non-zero entries in the resulting matrix. Half of these are $+1$ and the remaining half are 
	$-1$. Note that these entries depend only on the indices where the four permutations map the indices $i,j,i',j'$ and not on 
	the value of $n$ or where these permutations map the remaining indices.
	\qed
\end{proof}

\begin{lemma}\label{lem:szeroins}
	For any $P^{[2]}_{\sigma}\in S_0$, $P^{[2]}_{\sigma}-P^{[2]}_{\sigma'}$ lies in the linear hull of $S$ for every neighbor
	$P^{[2]}_{\sigma'}\in S_0$, provided $m\geqslant7$.
\end{lemma}
\begin{proof}
	Let $\sigma_1,\sigma_2,\sigma_3,\sigma_4$ be as defined in Lemma \ref{lem:identity2}. Let $\sigma_5,\sigma_6,\sigma_7,
	\sigma_8$ be four permutations different from $\sigma_1,\sigma_2,\sigma_3,\sigma_4$ but related to each other just like 
	$\sigma_1,\sigma_2,\sigma_3,\sigma_4$ are. This means that $\sigma_6$ is obtained from $\sigma_5$ by the same transposition 
	that is used to obtain $\sigma_2$ from $\sigma_1$, $\sigma_7$ is obtained from $\sigma_6$ by the same transposition that is
	used to obtain $\sigma_3$ from $\sigma_2$, and $\sigma_8$ is obtained from $\sigma_7$ by the same transposition that is used
	to obtain $\sigma_4$ from $\sigma_3$. So from Lemma \ref{lem:identity2}, we have $P^{[2]}_{\sigma_1}-P^{[2]}_{\sigma_2}=(
	P^{[2]}_{\sigma_3}-P^{[2]}_{\sigma_4})+(P^{[2]}_{\sigma_5}-P^{[2]}_{\sigma_6})-(P^{[2]}_{\sigma_7}-P^{[2]}_{\sigma_8})$. 
	Let $\sigma_1=\sigma$ and $\sigma_2=\sigma'$. If we can find $\sigma_3,\ldots,\sigma_8$ as defined above such that the
	corresponding vertices lie in $S$, then we are done. Since $\sigma,\sigma'\in S_0$, we have some index $a,a\neq i_r,r\in[m]$ 
	such that $\sigma(a)=\sigma'(a)=j_r$. Consider the case when $a=i_p,p\in[m]$ and $\sigma(i_r)\neq j_p$. The case when 
	$a\neq i_p$ for any $p\in[m]$ is similar. We can swap $\sigma'(a)$ with $\sigma'(i_r)$ to get $\sigma_3(a)=\sigma'(i_r),
	\sigma_3(i_r)=j_r$ which clearly lies in $S$. We obtain $\sigma_4$ from $\sigma_3$ by the transposition defined in Lemma 
	\ref{lem:identity2} and clearly $\sigma_4$ also lies in $S$. Next we select a permutation $\sigma_5\in S$ that matches with
	$\sigma_4$ at the four indices defined in Lemma \ref{lem:identity2} and also maps an index $i_{r'},r\neq r'$ to $j_{r'}$.
	So we have $\sigma_3,\sigma_4\in S_1,\sigma_5\in S_2$. Obtaining $\sigma_6,\sigma_7,\sigma_8$ as outlined above, we have
	$\sigma_6\in S_2,\sigma_7,\sigma_8\in S_1$. Next, consider the case when $a=i_p,p\in[m]$ and $\sigma(i_r)=j_p$. This can
	happen when $m$ is even and any transposition of $\sigma$ either results in a permutation in $S_0$ or in $S_2$. It is not
	possible to get a permutation in $S_1$ by a single transposition of $\sigma$. So we obtain $\sigma_3,\sigma_4$ as before
	but this time the vertices lie in $S_2$ instead of $S_1$. Moreover, this time we select a permutation $\sigma_5\in S$ that 
	matches with $\sigma_4$ at the four indices defined in Lemma \ref{lem:identity2} but has a pair of indices $i_x,i_y$ such 
	that $\sigma_5(i_x)=j_y$ but $\sigma_5(i_y)\neq j_x$. Obtaining $\sigma_6,\sigma_7,\sigma_8$ as above, we have $\sigma_5,
	\sigma_6\in S_2,\sigma_7,\sigma_8\in S_0$. We can now repeat the above argument with $\sigma=\sigma_7,\sigma'=\sigma_8$.
	This time however, by the choice of $\sigma_5\in S$, we have ensured that we are in the first case where we could express
	the difference vector as a combination of vertices only in $S$. Note that we need $m$ to be at least $7$ for the above
	argument to work. This is so because the transposition of $\sigma$ that gives $\sigma'$ can use upto four indices so that
	these indices are no longer available to get to $S$. Further, as in the second case above, two other indices swap with
	each other to get to $S_2$. So that takes up a total of six indices that are not available to get the desired $\sigma_5$.
	Now if $m$ is at least $7$, we are guaranteed to find a pair of indices $i_x,i_y$ such that $\sigma_5(i_x)=j_y$ but 
	$\sigma_5(i_y)\neq j_x$.
	\qed
\end{proof}

\begin{theorem}\label{thm:newfacets}
	The following inequality:
	\begin{equation*}
	\sum_{r=1}^{m}Y_{i_rj_r,i_rj_r}-\sum_{r < s}Y_{i_rj_r,i_sj_s}\le 1
	\end{equation*}	
	where $i_1,\dots,i_m$ are all distinct and $j_1,\dots,j_m$ are also distinct, is facet-defining for the QAP-polytope when 
	$m,n\geqslant 7$.
\end{theorem}
\begin{proof}
	From Lemma \ref{lem:skasnxt4} and Lemma \ref{lem:s3ss0} we can conclude that any vertex in $V\setminus(S\cup S_0)$ can be 
	expressed as a linear combination of the vertices in $S\cup S_0$. What remains to be shown is that any vertex in $S_0$
	can be expressed as a linear combination of the vertices in $S$ and a fixed vertex $\sigma^*\in S_0$. From Lemma 
	\ref{lem:szeroconn} we know that it is possible to go from any vertex in $S_0$ to any other vertex in $S_0$ via 
	transpositions such that all the intermediate vertices are in $S_0$. Let us fix some arbitrary vertex in $S_0$ as $\sigma^*$.
	So there is a path from every other vertex in $S_0$ to $\sigma^*$. Consider a vertex $\sigma\in S_0$. Let $\sigma,\sigma_1,
	\ldots,\sigma_t,\sigma^*$ be a path from $\sigma$ to $\sigma^*$. From Lemma \ref{lem:szeroins} we can express the difference
	of any vertex $\sigma\in S_0$ with any other vertex $\sigma'\in S_0$ such that $\sigma'$ is a transposition of $\sigma$, as
	a linear combination of the vertices in $S$. So we have $\sigma-\sigma_1\in span(S),\sigma_1-\sigma_2\in span(S),\ldots,
	\sigma_t-\sigma^*\in span(S)$ which implies that $\sigma-\sigma^*\in span(S)$ or $\sigma\in span(S\cup\{\sigma^*\})$.
	\qed
\end{proof}

\section{Membership testing}\label{sec:sepncomp}
In this section we consider the membership testing problem for each of the QAP relaxations defined in Section \ref{sec:relaxations}. That is, for each of these relaxations we wish to test whether a given point $x$ satisfies all the constraints. Note that the separation problem where one wishes to identify a violated inequality in case the answer to membership testing is negative is a harder problem. Typically for efficient use in cutting plane methods one would like to solve the (harder) separation problem.

We show that membership testing for inequalities QAP\ref{qap1}, QAP\ref{qap2}, or QAP\ref{qap4} is coNP-complete.

Recall that QAP\ref{qap1} is defined by the following set of inequalities:
\[\begin{tabular}{rcl}
$\displaystyle\sum_{r=1}^{m}Y_{i_rj_r,kl} - Y_{kl,kl}- \sum_{r < s}Y_{i_rj_r,i_sj_s}$ & $\leqslant$ & $0$
\end{tabular}\]
where $i_1,\dots,i_m,k$ are all distinct and $j_1,\dots,j_m,l$ are also
distinct. In addition, $n\geqslant 6, m\geqslant 3$.

\begin{theorem}\label{thm:sepqap1}
	Given a point $x\in \BR^{n^4}$ with $0\leqslant x \leqslant 1$, it is coNP-complete to decide whether $x$ satisfies all inequalities of QAP\ref{qap1}.
\end{theorem}
\begin{proof}
	The problem is clearly in coNP since given a violated inequality it can be checked quickly that it is indeed violated. 
	
	For establishing NP-hardness we will reduce the max-clique problem to membership testing for QAP\ref{qap1}. 
	Let the given instance of the max-clique problem be $G=(V,E)$ where $V=[n]$. We construct a $n$-partite graph $G'=(V',E')$ 
	where $V'=\{(ij)\}$ for $i,j\in [n]$ and $\{i_1j_1,i_2j_2\}\in E'$ if and only if $\{i_1,i_2\}\in E$. So if there is an
	edge $\{i,j\}\in E$ then we get a complete bi-partite graph between the partitions $i$ and $j$, else there is no edge between
	these two partitions. Consider a clique $C=\{i_1,i_2,\ldots,i_k\}$ of size $k$ in $G$. Then the set of vertices $\{i_1j_1,
	i_2j_2,\ldots,i_kj_k\}$ where $j_r$ could be any arbitrary index in $[n]$, forms a clique of size $k$ in $G'$. Conversely, 
	given a clique $C=\{i_1j_1,i_2j_2,\ldots,i_kj_k\}$ of size $k$ in $G'$, the set of vertices $\{i_1,i_2,\ldots,i_k\}$ forms a 
	clique of size $k$ in $G$. 
	Fix a pair of indices $k,l\in[n]$ arbitrarily and add edges $\{kl,i_rj_r\}$ for all $i_r\in\{[n]\setminus\{k\}\}$
	to $G'$ (if the edge is not already present). Now construct a point $Y$ as follows:
	\[
	Y_{i_1j_1,i_2j_2}= 
	\begin{cases}
	0,& \text{if } (kl)\notin\{(i_1j_1),(i_2j_2)\}\text{ and }\{i_1j_1,i_2j_2\}\in E'\\
	n,& \text{if } (kl)\notin\{(i_1j_1),(i_2j_2)\}\text{ and }\{i_1j_1,i_2j_2\}\notin E'\\
	1,& \text{if } (i_1\neq i_2=k\text{ and } j_1\neq j_2=l)\\
	& \text{ or }(k=i_1\neq i_2\text{ and }l=j_1\neq j_2)\\
	t,& \text{if } i_1=i_2=k\text{ and } j_1=j_2=l\\
	n^2,& \text{if } i_1=i_2\text{ and } j_1=j_2\text{ and } (i_1\neq k\text{ or } j_1\neq l)
	\end{cases}
	\]
	where $t\geqslant 2$ is a natural number. Notice that any point $Y$ satisfies all the inequalities of QAP\ref{qap1} if and 
	only if $\alpha Y$ satisfies them for all $\alpha\geqslant 0$. Therefore $Y$ can be scaled to satisfy $0\leqslant Y\leqslant 1$. We will ignore this scale factor and continue our argument with $Y$ as constructed above to avoid cluttered equations.
	
	We claim that $Y$ satisfies all the inequalities of QAP\ref{qap1} if and only if every clique in the subgraph induced
	by the neighborhood of the vertex $(kl)$ in $G'$, has size at most $t$.
	
	Suppose that the largest clique $C=\{i_1j_1,i_2j_2,\ldots,i_{t'}j_{t'}\}$ such that $\{i_rj_r,kl\ab\}\ab\in E'$ for $r\in[t']$, has 
	size $t'>t$. Without loss of generality, we can assume  that $i_1,\ldots,i_{t'}$ as well as $j_1,\ldots,j_{t'}$ are distinct. Consider the inequality $h_C$ corresponding to the choice of indices  $\{i_1,\ldots,i_{t'},k\}, \{j_1,\ldots,j_{t'},l\}.$ From the above construction, $\sum_{r=1}^{t'}Y_{i_rj_r,kl}=t',$ $Y_{kl,kl}=t$ and $\sum_{r < s,r,s\in[t']}
	Y_{i_rj_r,i_sj_s}=0$ giving $t'\leqslant t$ and $Y$ violates $h_C$.
	
	Now suppose that every clique $C$ in the subgraph induced by the neighborhood of $(kl)$, has size at most $t$ and there exists
	an inequality of QAP\ref{qap1} defined by the sets $\{i_1,\ldots,i_m,k'\}, \{j_1,\ldots,j_m,l'\}$ that is violated by the above point $Y$. 
	Notice that $k'=k$ and  $l'=l$ must hold. If not then $\sum_{r=1}^{m}Y_{i_rj_r,k'l'}\ab -\sum_{{r < s,r}}Y_{i_rj_r,i_sj_s}\leqslant nm \leqslant n^2 = Y_{k'l',k'l'}$ and $Y$ does not violate the inequality.
	So any violated inequality must have $\sum_{r=1}^{m}Y_{i_rj_r,kl}>t+\sum_{r < s}Y_{i_rj_r,i_sj_s}$. Since $Y_{i_rj_r,kl}=1$ 
	for all $i_r\in\{[n]\setminus\{k\}\}$, we have 
	$m>t+\sum_{r < s}Y_{i_rj_r,i_sj_s}$ for any violated inequality, which is not possible if $Y_{i_rj_r,i_sj_s}=n$ since $m\leqslant n$ and $t\geqslant 2$. Therefore, $Y_{i_rj_r,i_sj_s}=0$ for all distinct $r,s\in[m]$  and so $\{i_rj_r,,i_sj_s\}\in E'$. But then, $m>t$ giving a clique of size larger than $t$ in the neighborhood of $(kl)$ contradicting the assumption that the every such clique has size at most $t$.
	
	Therefore, given a membership oracle for QAP\ref{qap1} we can compute the size of the largest clique in any graph except $K_n$ by calling such an oracle for various choices of $k,l$ and $t$ and outputting the largest value of $t$ for which the above constructed point satisfies all the inequalities.
	
	%
	\qed
\end{proof}

\medskip

Next, recall that QAP\ref{qap2} is defined by the following set of inequalities:
\[\begin{tabular}{rcl}
$\displaystyle(\beta-1)\sum_{(ij)\in P\times Q}Y_{ij,ij}-\sum_{\substack{(ij),(kl)\in P\times Q\\i<k}}Y_{ij,kl}$ & $\leqslant$ & $\frac{\beta^2-\beta}{2}$
\end{tabular}\]
where $P,Q\subset[n]$. In addition (i) $\beta +1 \leqslant\lvert P\rvert,\lvert Q\rvert\leqslant n-3$, (ii) $\lvert P\rvert+\lvert Q\rvert\leqslant 
n-3+\beta$, (iii) $\beta\geqslant2$.

\begin{theorem}\label{thm:sepqap2}
	Given a point $x\in \BR^{n^4}$ with $x\geqslant 0$, it is coNP-complete to decide whether $x$ satisfies all inequalities of QAP\ref{qap2}. 
\end{theorem}
\begin{proof}
	Again we will reduce the max-clique problem to membership testing for QAP\ref{qap2}. Given an instance of the max-clique 
	problem, $G=(V,E)$ with $|V|=n$, we construct a point $Y$ as follows:
	\[
	Y_{i_1j_1,i_2j_2}=
	\begin{cases}
	0,& \text{if } \{i_1,i_2\}\in E\\
	n^2,& \text{if } \{i_1,i_2\}\notin E\text{ } (i_1\neq i_2)\\
	1/t,& \text{if } i_1=i_2\text{ and } j_1=j_2=1\\
	0,& \text{if } i_1=i_2\text{ and } j_1=j_2\text{ and } j_1>1
	\end{cases}
	\]
	where $1\leqslant t\leqslant n-4$ is some natural number. We claim that $Y$ satisfies all inequalities of QAP\ref{qap2} if and only if $G$ doesn't contain a clique of size larger than $t$. This gives an algorithm to find the size of largest clique in $G$ by increasing $t$ gradually and computing the smallest value of $t$ for which $Y$ becomes feasible. If $Y$ remains infeasible for $t=n-4$ then the largest clique in $G$ has size $n-3$ or more. This can be determined by checking the $O(n^3)$ possible subsets of vertices of $G$.
	
	To prove the claim, let $P$ be a clique in $G$ with $t<|P|\leqslant n-3$. Define $Q=\{1\}, \beta=2$ and consider the inequality $h$ defined by $P,Q,\beta$. It can be checked that $Y$ violates $h$.
	
	\medskip
	Conversely, let $h$ be an inequality defined by $P,Q,\beta$ that is violated by $Y$. 
	
	We first observe that $1\in Q$. Suppose not, then for any $(i,j)\in P\times Q$ we have $Y_{ij,ij}=0$. Since $\beta^2-\beta \geqslant 0$ for all natural $\beta\geqslant 2$, $h$ cannot be violated by $Y$.
	It follows that if $i,k\in P$ and $i\neq k$, then $\{i,k\}\in E.$
	Again, suppose not. Then, $Y_{ij,kl}=n^2$ for all $j,l$ but $(\beta-1)\sum_{(ij)\in P\times Q}Y_{ij,ij}\leqslant(n-4)\cdot(n-3) < n^2$. So $Y$ cannot violate $h$ as $(\beta-1)\sum_{(ij)\in P\times Q}Y_{ij,ij}-\sum_{(ij),(kl)\in P\times Q,i<k}Y_{ij,kl}$ is negative but $\beta^2-\beta$ is nonnegative.
	
	So we have that if $Y$ violates $h$ then $Y_{ij,kl}=0$ for all $i\neq k$. Further, such an inequality must have $\{i,k\}\in E$ for all distinct $i,k\in P$. That is $P$ must form a clique in $G$.
	Recall that only $1\in Q$ contributes a non-zero value to the left hand side expression of $h$. Therefore, if $Y$ violates $h$ then $|P|/t=\sum_{(ij)\in P\times Q}Y_{ij,ij}>\beta/2$ and $G$ contains a clique of size larger than $t\beta/2$, that is, larger than $t.$     
	\qed
\end{proof}

\medskip

Finally, recall that QAP\ref{qap4} is defined by the following set of inequalities:
\[\begin{tabular}{rcl}
$\displaystyle\sum_{r=1}^{m}Y_{i_rj_r,i_rj_r}-\sum_{r < s}Y_{i_rj_r,i_sj_s}$ & $\leqslant$ & $1$
\end{tabular}\]
where $i_1,\dots,i_m$ are all distinct and $j_1,\dots,j_m$ are also
distinct. In addition, $m,n\geqslant 7$.

\begin{theorem}\label{thm:sepqap4}
	Given a point $x\in \BR^{n^4}$ with $x\geqslant 0$, it is coNP-complete to decide whether $x$ satisfies all inequalities of QAP\ref{qap4}.
\end{theorem}
\begin{proof}
	Given an instance of the max-clique problem, $G=(V,E)$ with $|V|=n$, we construct a point $Y$ as follows:
	\[
	Y_{i_1j_1,i_2j_2}= 
	\begin{cases}
	0,& \text{if } \{i_1,i_2\}\in E\\
	n/6,& \text{if } i_1\neq i_2\text{ and }\{i_1,i_2\}\notin E\\
	1/t,& \text{if } i_1=i_2\text{ and } j_1=j_2\\
	0,& \text{otherwise}
	\end{cases}
	\]
	where, $t\geqslant 6$ is a natural number.
	
	We claim that $Y$ is infeasible if and only if there exists a clique in $G$ of size at least $t+1$.
	
	Suppose there exists a clique $C=\{p_1,\ldots,p_m\}$ in $G$ with $m\geqslant t+1$. Consider the inequality $h_C$ defined by the indices $i_1,\ldots,i_m$ and $j_1,\ldots,j_m$ with $i_k=j_k=p_k$ for all $k\in[m]$. Then, $Y$ violates $h_C$ because $m/t>1$.
	
	Conversely, suppose every clique in $G$ has size at most $t$ and let $h$ be a violated inequality defined by indices $i_1,\ldots,i_m$ and $j_1,\ldots,j_m$. It must hold that $\{i_r,i_s\}\in E$ otherwise the left hand side in the inequality $h$ with respect to $Y$ is at most $m/t-n/6$ which is at most zero since $t\geqslant 6$ and $m\leqslant n$ and so $Y$ cannot violate $h$. So for a violation, $G$ must contain a clique of size $m$. But then $Y_{i_rj_r,i_sj_s}=0$ for all distinct $r,s\in[m]$ and so $m/t>1$ contradicting the assumption that $G$ contains no cliques of size larger than $t$.
	
	Therefore, given a graph $G=(V,E)$ if $Y$ is feasible for all values of $t\geqslant 6$ then the size of a largest clique is at most $6$ and can be computed in polynomial time. Otherwise, the largest value of $t$ for which $Y$ is infeasible equals the size of the largest clique in $G$ minus one.
	\qed
\end{proof}

\section{Extension Complexity}\label{sec:extncomp}
In this section we will prove that any relaxation of the QAP-polytope for which any of the families of inequalities defined in Section \ref{sec:relaxations} are valid, has superpolynomial extension complexity. A set of linear inequalities is said to be a relaxation of the QAP-polytope if it contains the QAP-polytope. 


We will use Theorem \ref{thm:extformcc} to prove each of the results in this section. The general argument would be as follows. First we will show that certain specific matrices require exchange of many bits in any EF-protocol for them. Then, supposing that $\qaprelax{}{n}$ is any  relaxation of \qappolytope{n} for which inequalities $\cH$ are valid, we will prove that any EF-protocol for the slack matrix of $\qaprelax{}{n}$ requires an exchange of many bits by showing that computing the slack entries corresponding to the valid inequalities in $\cH$ and the vertices of \qappolytope{n} allows us to compute our special matrices by exchanging few extra bits. So the slack matrix of $\qaprelax{}{n}$ contains a submatrix that is difficult to compute and thus the whole slack matrix is difficult to compute.

\subsection{Hard matrices for EF-protocols}
Let $n$ be natural numbers. Consider the $2^n\times 2^n$ sized matrices $M^1_n, N^0_n, N^1_n$ defined as follows. The rows and columns of these matrices are indexed by $0/1$ vectors of length $n$. Let $a, b \in \{0,1\}^n$. The entry with row index $a$ and column index $b$ for each of these matrices is defined as follows:
\[\begin{tabular}{rcl}
	$M^k_n(a,b)$ & $:=$ & $(a^\intercal b-k)^2$ \\
	$N^k_n(a,b)$ & $:=$ & $(a^\intercal b-k)\cdot(a^\intercal b-k-1)$
\end{tabular}\]

The matrices $M^1_n$ require $\Omega(n)$ bits to be exchanged in any EF-protocol and were used to prove that the correlation polytope has superpolynomial extension complexity by Fiorini et al. \cite{jacm/FioriniMPTW15}. Later Kaibel and Weltge \cite{dcg/KaibelW15} gave a simple combinatorial proof of the hardness of computing the same matrices $M^1_n$.  In particular, the following holds:

\begin{lemma}\cite{jacm/FioriniMPTW15,dcg/KaibelW15}\label{lem:disjointnesshard}
Any EF-protocol for computing $M^1_n$ requies an exchange of $\Omega(n)$ bits.
\end{lemma}

For our purposes we need to work with matrices $N^k_n$ which we now show to require $\Omega(n-k)$  bits for $k\geqslant 1$ in any EF-protocol as well.

\begin{lemma}\label{lem:n0}
There exists an EF-protocol for computing $N^0_n$ that requires an exchange of $2\lceil\log{n}\rceil$ bits.	
\end{lemma}
\begin{proof}
	Alice and Bob have $a,b\in\{0,1\}^n$ respectively and they wish to output $a^\intercal b\cdot (a^\intercal b-1)$ in expectation. Alice selects two indices $i\neq j\in [n]$ uniformly at random. That is, the probability of any two indices being chosen is $1/\binom{n}{2}$. If $a_i=0$ or $a_j=0$, Alice outputs zero and the protocol ends. Otherwise $a_i=a_j=1$ and Alice sends the binary encoding of indices $i,j$ to Bob using $2\lceil\log{n}\rceil$ bits. If $b_i=b_j=1$, Bob outputs $n(n-1)$ otherwise he outputs zero, and the protocol ends.
	
	A non-zero value is output if and only if the indices selected by Alice satisfy $a_i=a_j=b_i=b_j=1$. This happens with probability $\binom{a^\intercal b}{2}/\binom{n}{2}$. Therefore the expected output is $n(n-1)\cdot\frac{\binom{a^\intercal b}{2}}{\binom{n}{2}}=N^0_n(a,b).$ 
	\qed
\end{proof}

\begin{lemma}\label{lem:nmatrixhard}
	Any EF-protocol for computing $N^k_n$ for $k\geqslant 1$ requires $\Omega(n-k)$ bits to be exchanged.
\end{lemma}
\begin{proof}
	We first prove the result for the case of $k=1$.
	
	For any real number $\alpha$ it holds that $\alpha(\alpha-1)+(\alpha-1)(\alpha-2)=2(\alpha-1)^2$. Using $\alpha=a^\intercal b$ we get that for all $a,b\in\{0,1\}^n$ we have that $M^1_n(a,b)=\frac{N^0_n(a,b)+N^1_n(a,b)}{2}$. 
	
	Now, suppose that there exists an EF-protocol for computing $N^1_n$ that requires $c$ bits to be exchanged. To compute $M^1_n$, Alice decides whether to compute $N^0_n$ or $N^1_n$ (by tossing a fair coin) and tells this to Bob using one bit. They then proceed with the appropriate protocol. The expected value is clearly $\frac{N^0_n(a,b)+N^1_n(a,b)}{2}$ which equals $M^1_n(a,b)$ and they need to exchange only $1+\max\{\ab 2\lceil\log{n}\rceil, c\}$ bits. From Lemma \ref{lem:disjointnesshard} we know that $\Omega(n)$ bits may need to be exchanged for this. Therefore, $c=\Omega(n)$ and the lemma holds for $k=1$.	
	
	Now, let $a,b\in\{0,1\}^{n-k+1}$ be two binary strings of length $n-k+1$ for some $k\geqslant 2$. Obtain $a',b'\in\{0,1\}^{n}$ by appending $k-1$ ones to both. Then, $a'^\intercal b'=a^\intercal b + k - 1$ and so $(a'^\intercal b'-k)(a'^\intercal b'-k-1)=(a^\intercal b-1)(a^\intercal b-2)$. Hence, $N^1_{n-k+1}$ is a submatrix of $N^k_{n}$. It then follows that any EF-protocol for $N^k_{n}$ requires $\Omega(n-k+1)$ bits to be exchanged thus proving the lemma for all $k\geqslant 1$.
	\qed
\end{proof}

\subsection{Relaxations of \qappolytope{n}}
Recall the inequalities (\ref{qap5ineq}) of QAP\ref{qap5}:
\[\begin{tabular}{rcl}
$\displaystyle\sum_{ijkl}n_{ij}n_{kl}Y_{ij,kl} - (2\beta-1)\sum_{ij}n_{ij}Y_{ij,ij}$ & $\geqslant$ & $\frac{1}{4}-(\beta-\frac{1}{2})^2$
\end{tabular}\]

As noted in Section \ref{sec:relaxations} these are the most general family of valid inequalities for the QAP-polytope. Therefore any lower bound on the extension complexity of a relaxation \qaprelax{}{n} corresponding to any of the families QAP\ref{qap1}-QAP\ref{qap4} also hold for QAP\ref{qap5} and so we state our next theorem without proof.

\begin{theorem}\label{thm:extncompqap5}
	Let $\qaprelax{\ref{qap5}}{n}$ be any bounded relaxation of the QAP-polytope such that the  inequalities of QAP\ref{qap5} are valid for $\qaprelax{\ref{qap5}}{n}$. Then $\xc(\qaprelax{\ref{qap5}}{n})\geqslant 2^{\Omega(n)}$.
\end{theorem}

Now consider any relaxation of \qappolytope{n} for which the inequalities of QAP\ref{qap1} are valid. 

\begin{lemma}\label{lem:slack:qap1}
	Let $i_1,i_2,\ldots,i_m,k$ be distinct indices in $[n]$, and $j_1,j_2,\ldots,j_m,l$ be all distinct 
	indices in $[n]$ as well. Let $\sigma\in S_n$ be such that $q$ indices $(i_r,j_r)$ satisfy $P_{\sigma}(i_r,j_r)=1$. Then, the slack of $\displaystyle\sum_{r=1}^{m}Y_{i_rj_r,kl} - Y_{kl,kl}- \sum_{r < s}Y_{i_rj_r,i_sj_s} \leqslant0$ with respect to $P^{[2]}_{\sigma}$ is $\binom{q-P_\sigma(k,l)}{2}.$
\end{lemma}
\begin{proof}
The slack is $P_\sigma(k,l)+\binom{q}{2}-qP_\sigma(k,l)$. If $P_\sigma(k,l)=0$ then this equals $\binom{q}{2}$. If $P_\sigma(k,l)=1$ then this equals $1+\binom{q}{2}-q=\frac{2+q(q-1)-2q}{2}=\binom{q-1}{2}.$
\qed
\end{proof}

\begin{theorem}\label{thm:extncompqap1}
	Let $\qaprelax{\ref{qap1}}{n}$ be any bounded relaxation of \qappolytope{n} such that the inequalities of QAP\ref{qap1} are valid for $\qaprelax{\ref{qap1}}{n}$. Then $\xc(\qaprelax{\ref{qap1}}{n})\geqslant2^{\Omega(n)}$.
\end{theorem}
\begin{proof}
	Suppose there exists an EF-protocol for computing the slack matrix of $\qaprelax{\ref{qap1}}{n}$ that requires at most $c$ bits to be exchanged. That is if Alice is given any valid inequality for $\qaprelax{\ref{qap1}}{n}$ and Bob any feasible point in $\qaprelax{\ref{qap1}}{n}$ they can compute the corresponding slack in expectation by exchanging at most $c$ bits. We will show that they can modify this EF-protocol to get an EF-protocol for matrix $N^1_n$ with at most $O(c+\log n)$ bits exchanged. By Lemma \ref{lem:nmatrixhard} this requires $\Omega(n)$ bits to be exchanged and so $c=\Omega(n)$. Finally, applying Theorem \ref{thm:extformcc} we will get that $\xc(\qaprelax{\ref{qap1}}{n})\geqslant2^{\Omega(n)}$.

	So suppose, Alice and Bob get $a,b\in\{0,1\}^n$ respectively and wish to compute $N^1_n(a,b)=(a^\intercal b-1)(a^\intercal b-2)$ in expectation. We can assume that Alice receives neither the all zero nor the all one vector. If $a=(0,\ldots,0)$ then she can output zero and stop. If $a=(1,\ldots,1)$ then she  can tell Bob this using one bit and Bob can output the number of nonzero entries in $b$. Further, we can assume that the vector $b$ contains at least three zero entries. Otherwise Bob can tell Alice using at most $2\log n$ bits the indices where $b$ is zero and Alice can output the correct value.
	
	Let $p_1,\ldots,p_m$ be the indices where $a$ is non-zero and let $p$ be an arbitrary index such that $a_p=0$. Alice creates the inequality corresponding to sets $i_1,\ldots,i_m,k$ and $j_1,\ldots,j_m,l$ where $i_1{=}j_1{=}p_1,\ldots,i_m{=}j_m{=}p_m$ and $k=l=p$. Alice then sends the index $p$ to Bob who sets $b_p=1$ if it is not already so. Bob then creates any permutation $\sigma_b$ such that $\sigma_b(i)=i$ if $b_i=1$ and $\sigma_b(i)\neq i$ if $b_i=0$. This is clearly possible since $b$ still contains at least two zeroes. Bob selects the vertex $P^{[2]}_{\sigma_b}$ of \qappolytope{n} corresponding to this permutation. Clearly, $a^\intercal b$ equals the number of index pairs $(i_r,j_r)$ in the set created by Alice such that $P^{[2]}_{\sigma_b}(i_r,j_r)=1$. 

	Using $P_{\sigma}(k,l)=P_{\sigma_b}(p,p)=1$  and
	$q=a^\intercal b$ in Lemma \ref{lem:slack:qap1} we see that the slack of Alice's inequality with respect to Bob's vertex of \qappolytope{n} is exactly $\binom{a^\intercal b -1}{2}=\frac{1}{2}N^1_n(a,b)$ and hence they can just use the protocol for computing the slack matrix of $\qaprelax{\ref{qap1}}{n}$ for computing $N^1_n$ by agreeing that every time they wish to output something they would output twice as much. 
	\qed
\end{proof}

\medskip

Next, consider any relaxation of \qappolytope{n} for which the inequalities of QAP\ref{qap2} are valid. Recall the inequalities (\ref{qap2ineq}) of QAP\ref{qap2}:
\[\begin{tabular}{rcl}
$\displaystyle(\beta-1)\sum_{(ij)\in P\times Q}Y_{ij,ij}-\sum_{\substack{(ij),(kl)\in P\times Q\\i<k}}Y_{ij,kl}$ & $\leqslant$ & $\frac{\beta^2-\beta}{2}$
\end{tabular}\]
\begin{lemma}\label{lem:slack:qap2}
	Let $h$ be an inequality in QAP\ref{qap2} given by sets $P,Q \subset [n]$ and $\beta\in\BN$. Let $\sigma\in S_n$ such that $q$ index pairs $(i,j)$ in $P\times Q$ satisfy $P_\sigma(i,j)=1$. Then, the slack of $h$ with respect to $P^{[2]}_\sigma$ is $\binom{q-(\beta-1)}{2}.$
\end{lemma}
\begin{proof}
The slack of $h$ with respect to $P^{[2]}_\sigma$ is 
\begin{align}
&~ \frac{\beta^2-\beta}{2}-(\beta-1)q+\binom{q}{2}\nonumber\\
= &~ \frac{\beta(\beta-1)-2q(\beta-1)+q(q-1)}{2}\nonumber\\
= &~ \frac{(\beta-1)(\beta-q)+q(q-\beta)}{2}\nonumber\\
= &~ \frac{(q-\beta)(q-\beta+1)}{2}\nonumber\\
= &~ \binom{q-(\beta-1)}{2}
\end{align}
\qed
\end{proof}
\begin{theorem}\label{thm:extncompqap2}
		Let $\qaprelax{\ref{qap2}}{n}$ be any bounded relaxation of \qappolytope{n} such that the inequalities of QAP\ref{qap2} are valid for $\qaprelax{\ref{qap2}}{n}$. Then $\xc(\qaprelax{\ref{qap2}}{n})\geqslant2^{\Omega(n)}$.
\end{theorem}
\begin{proof} 

	Suppose Alice and Bob get $a,b\in\{0,1\}^m$ respectively and wish to compute $N^1_m(a,b)$ in expectation. 
	Let $n=2m+1$ and let $p_1,\ldots,p_t$ be the indices where $a$ is non-zero. Alice creates the inequality corresponding to sets $P=Q=\{p_1,\ldots,
	p_t\}$ and $\beta=2$. Note that $|P|\leqslant m=(n-1)/2$ so $|P|+|Q|=2|P|\leqslant n-1=n-3+\beta$. So the value of $\beta=2$
	the choices of $P,Q,\beta$ satisfy the conditions of QAP\ref{qap2} for \qappolytope{n}.
	
	Bob creates the following permutation $\sigma_b\in S_{n}$ 
	\[\sigma_b(i)=
	\begin{cases}
		i & \quad \text{ if } 1\leqslant i \leqslant m \text{ and } b_i=1,\\
		i+m & \quad \text{ if } 1\leqslant i \leqslant m \text{ and } b_i=0,\\
		i-m & \quad \text{ if } m<i \leqslant 2m+1 \text{ and } b_i=0,\\
		i & \quad \text{ otherwise. }
	\end{cases}
	\]
	and selects the vertex $P^{[2]}_{\sigma_b}$ of the \qappolytope{n}. Since $\qaprelax{\ref{qap2}}{n}$ is a relaxation of \qappolytope{n}, Bob has picked a feasible point of $\qaprelax{\ref{qap2}}{n}$. Notice that $a^\intercal b$ equals the number of index pairs $(i,j)$ in the set $P\times Q$ such that $P_{\sigma_b}(i,j)=1$. 
	
	Therefore using $q=a^\intercal b$ and $\beta=2$ in Lemma \ref{lem:slack:qap2} we get that the slack of Alice's inequality with respect to Bob's feasible point of $\qaprelax{\ref{qap2}}{n}$ is ${a^\intercal b-(\beta-1)\choose 2}=\frac{a^\intercal b-1)(a^\intercal b-2)}{2}=\frac{1}{2}N^1_m(a,b)$
	and hence they can just use the protocol for 
	computing the slack matrix of $\qaprelax{\ref{qap2}}{n}$ for computing $N^1_m$ by agreeing that every time they wish to output something they 
	would output twice as much. 
	\qed
\end{proof}

We would like to remark that in the proof of the previous Theorem one could have as easily used a value of $\beta$ other than two as long as $\beta\leqslant \alpha n$ for some $0 <\alpha <1$. So the extension complexity for QAP\ref{qap2} remains superpolynomial even for any fixed value of $\beta$ if $\beta\leqslant\alpha n$ for some $0<\alpha<1$.

\medskip

Next, consider any relaxation of \qappolytope{n} for which the inequalities of QAP\ref{qap3} are valid. Recall the inequalities (\ref{qap3ineq}) of QAP\ref{qap3}:
\[\begin{tabular}{rrrrcl}
$-$ & $\displaystyle(\beta-1)\sum_{(ij)\in P_1\times Q}Y_{ij,ij}$ & $+$ & $\displaystyle\beta \sum_{(ij)\in P_2\times Q}Y_{ij,ij} $ & &\\
$+$ & $\displaystyle\sum_{\substack{i<k \\ (ij)\in P_1\times Q\\ (kl)\in P_1\times Q}}Y_{ij,kl}$ & $+$ &$\displaystyle\sum_{\substack{i<k \\ (ij)\in P_2\times Q\\(kl)\in P_2\times Q}} Y_{ij,kl}$ & & \\
& & $-$ & $\displaystyle\sum_{\substack{(ij)\in P_1\times Q \\ (kl)\in P_2\times Q}}Y_{ij,kl}$ & $\geqslant$ & $ \frac{\beta-\beta^2}{2}$
\end{tabular}\]

\begin{lemma}\label{lem:slack:qap3}
	Consider an inequality $h$ in QAP\ref{qap3} given by sets $P_1,P_2,Q \subset [n]$ and $\beta\in\BN$. Let $\sigma\in S_n$ be such that $q_1$ index pairs $(i,j)$ in $P_1\times Q$ satisfy $P_\sigma(i,j)=1$ and $q_2$ index pairs $(i,j)$ in $P_2\times Q$ satisfy $P_\sigma(i,j)=1.$ Then, the slack of $h$ with respect to $P^{[2]}_\sigma$ is $\binom{q_1-(\beta-1)}{2}+\frac{2q_2\beta+q_2(q_2-1)-2q_1q_2}{2}.$
\end{lemma}
\begin{proof}
The slack of $h$ with respect to $P^{[2]}_\sigma$ is 
	\begin{align}
	& ~-(\beta-1)q_1 +\beta q_2+\binom{q_1}{2}+\binom{q_2}{2}-q_1q_2-\frac{\beta(1-\beta)}{2}\nonumber\\
	= &~ \frac{q_1(q_1-1)-q_1(\beta-1)-q_1(\beta-1)-\beta(1-\beta)}{2}+\frac{2q_2\beta+q_2(q_2-1)-2q_1q_2}{2}\nonumber\\
	= &~ \frac{q_1(q_1-\beta)-(\beta-1)(q_1-\beta)}{2}+\frac{2q_2\beta+q_2(q_2-1)-2q_1q_2}{2}\nonumber\\
	= &~ \frac{(q_1-\beta+1)(q_1-\beta)}{2}+\frac{2q_2\beta+q_2(q_2-1)-2q_1q_2}{2}\nonumber\\
	= &~ \binom{q_1-(\beta-1)}{2}+\frac{2q_2\beta+q_2(q_2-1)-2q_1q_2}{2}\label{slack:qap3}
	\end{align}	
\qed
\end{proof}

\begin{theorem}\label{thm:extncompqap3}
		Let $\qaprelax{\ref{qap3}}{n}$ be any bounded relaxation of \qappolytope{n} such that the inequalities of QAP\ref{qap3} are valid for $\qaprelax{\ref{qap3}}{n}$. Then $\xc(\qaprelax{\ref{qap3}}{n})\geqslant2^{\Omega(n)}$.
	\end{theorem}
\begin{proof}

	Suppose Alice and Bob get $a,b\in\{0,1\}^m$ respectively and wish to compute $N^1_m(a,b)$ in expectation. We can assume that $a$ is not all one vector. If not, she can send this information to Bob using a single bit who can output the number of ones in vector $b$. 
	Let $P_1=Q=\{i\mid a_i=1\}$ and let $P_2\subseteq [m]$ be such that $P_1\cap P_2=\emptyset$ and $|P_2|=1$. Clearly, $i\in P_2 \implies a_i=0$. 
	Alice creates the inequality corresponding to the sets $P_1,P_2,Q$ and $\beta=2$. One can verify that all the conditions of
	QAP\ref{qap3} are satisfied for the above choice of $P_1,P_2,Q$ and $\beta$. Alice tells Bob the set $P_2$ using $\log n$ bits.
	
	Bob creates the following permutation $\sigma_b\in S_{2m+1}$ 
\[\sigma_b(i)=
\begin{cases}
i & \quad \text{ if } i\in [m]\setminus P_2 \text{ and } b_i=1,\\
i+m & \quad \text{ if } i\in P_2,\\
i+m & \quad \text{ if } 1\leqslant i \leqslant m \text{ and } b_i=0,\\
i-m & \quad \text{ if } m<i \leqslant 2m+1 \text{ and } b_i=0,\\
i-m & \quad \text{ if } i-m\in P_2,\\
i & \quad \text{ otherwise. }
\end{cases}
\]
and selects the vertex $P^{[2]}_{\sigma_b}$ of the \qappolytope{n}. Since $\qaprelax{\ref{qap2}}{n}$ is a relaxation of \qappolytope{n}, Bob has picked a feasible point of $\qaprelax{\ref{qap2}}{n}$. Notice that $a^\intercal b$ equals the number of index pairs $(i,j)\in P_1\times Q$ such that $P_{\sigma_b}(i,j)=1$. Also, notice that the number of index pairs $(i,j)\in P_2\times Q$ such that $P_{\sigma_b}(i,j)=1$ is zero. 

Using $q_1=a^\intercal b, q_2=0$ and $\beta=2$ in Lemma \ref{lem:slack:qap3} we get that the slack of Alice's inequality with respect to Bob's point 
	of $\qaprelax{\ref{qap3}}{2m+1}$ is ${a^\intercal b-(\beta-1)\choose 2}={a^\intercal b-1\choose 2}=\frac{1}{2}N^1_{2m+1}(a,b)$. Therefore, Alice
	and Bob can just use the protocol for computing the slack matrix of $\qaprelax{\ref{qap3}}{2m+1}$ for computing $N^1_m(a,b)$ by agreeing that every 
	time they wish to output something they would output twice as much. 
	\qed
\end{proof}

Once again we remark that the previous proof can be easily modified for any other value of $\beta$ as long as $\beta\leqslant \alpha n$ for some $0< \alpha < 1$. So the extension complexity lower bound also holds for each subfamily of inequalities obtained for any such fixed $\beta$.

\medskip

Finally, consider any relaxation of \qappolytope{n} for which the inequalities of QAP\ref{qap4} are valid. Recall the inequalities (\ref{qap4ineq}) of QAP\ref{qap4}:
\[\begin{tabular}{rcl}
$\displaystyle\sum_{r=1}^{m}Y_{i_rj_r,i_rj_r}-\sum_{r < s}Y_{i_rj_r,i_sj_s}$ & $\leqslant$ & $1$
\end{tabular}\]

\begin{lemma}\label{lem:slack:qap4}
	Let $h$ be an inequality in QAP\ref{qap4} given by $i_1,\ldots,i_m$ all distinct and $j_1,\ldots,j_m$ all distinct. Let $\sigma\in S_n$ with $q$ index pairs $(i_r,j_r)$ such that $P_\sigma(i_r,j_r)=1.$ Then, the slack of $h$ with respect to $P^{[2]}_\sigma$ is $\binom{q-1}{2}.$
\end{lemma}
\begin{proof}
The slack equals $1-q+\binom{q}{2}= \frac{2-2q+q(q-1)}{2}=\frac{(q-1)(q-2)}{2}.$\qed
\end{proof}

\begin{theorem}\label{thm:extncompqap4}
		Let $\qaprelax{\ref{qap4}}{n}$ be any bounded relaxation of \qappolytope{n} such that the  inequalities of QAP\ref{qap4} are valid for $\qaprelax{\ref{qap4}}{n}$. Then $\xc(\qaprelax{\ref{qap4}}{n})\geqslant2^{\Omega(n)}$.
\end{theorem}
\begin{proof}

	Suppose Alice and Bob get $a,b\in\{0,1\}^n$ respectively and wish to compute $N^1_n(a,b)$ in expectation. We can assume that $n\geqslant 7$ otherwise they can exchange their entire vectors using six bits. Let $p_1,\ldots,
	p_m$ be the indices where $a$ is non-zero. We can assume that $m\geqslant 7$ otherwise Alice can send the indices where $a$ is non-zero using at most $6\log{n}$ bits and Bob can output the correct value of $N^1_n(a,b)$. So $m,n\geqslant 7$ and Alice creates the inequality corresponding to sets $i_1,\ldots,i_m$ and $j_1,
	\ldots,j_m$ where $i_1{=}j_1{=}p_1,\ldots,\ab i_m{=}j_m{=}p_m$. 
	
	We can also assume that $b$ contains at least two zeroes otherwise Bob can tell Alice the index where $b$ is zero using $\log{n}$ bits and Alice can output the correct value of $N^1_n(a,b)$. Now Bob creates the following permutation $\sigma_b\in S_n$ such that $\sigma_b(i)=i$ if $b_i=1$ and $\sigma_b(i)\neq i$ if $b_i=0$. This is clearly possible since $b$ contains at least two zeroes. Bob selects the vertex $P^{[2]}_{\sigma_b}$ of \qappolytope{n} corresponding to this permutation. Clearly, $a^\intercal b$ equals the number of index pairs $(i_r,j_r)$ in the inequality created by Alice such that $P^{[2]}_{\sigma_b}(i,i)=1$. 
	
	Setting $q=a^\intercal b$ in Lemma \ref{lem:slack:qap4} we get that the slack of Alice's inequality with respect to Bob's point of \qaprelax{\ref{qap4}}{n} is
	exactly $\frac{1}{2}N^1_n(a,b)$ and hence they can just use the protocol for computing the slack matrix of $\qaprelax{\ref{qap4}}{n}$ for computing 
	$N^1_n$ by agreeing that every time they wish to output something they would output twice as much. 
	\qed
\end{proof}

\subsection*{\ackname}
Pawan Aurora is partially supported by grant MTR/2018/000861 of the Science and Engineering Research Board, Government of India. 
Hans Raj Tiwary was partially supported by grant 17-09142S of GA\v{C}R.
{\bibliographystyle{spmpsci}      
\bibliography{qap-facets}}   

\end{document}